\documentclass[11pt]{article}

\usepackage{amsmath,amssymb,amsfonts}
\usepackage{latexsym}
\usepackage{subcaption}
\usepackage{graphicx}
\usepackage{tabu}
\usepackage[backref=page]{hyperref}
\usepackage{xcolor}
\usepackage{wrapfig}
\usepackage{setspace}
\usepackage{xspace}

\usepackage[linesnumbered,lined,ruled,noend]{algorithm2e}
\usepackage{color}
\providecommand{\email}[1]{\href{mailto:#1}{\nolinkurl{#1}\xspace}}
\newcommand{\pr}{\mathrm{pr}}
\definecolor{sapphire}{rgb}{0.03, 0.15, 0.4}
\newcommand{\norm}[1]{\left\lVert#1\right\rVert}
\newcommand{\br}[1]{\left\{#1\right\}}

\newcommand{\mbs}{\mathbb{S}}

\newcommand{\disum}{\displaystyle\sum}

\newcommand{\eps}{\varepsilon}

\newcommand{\cost}{\displaystyle\mathrm{cost}}

\newcommand{\opt}{\mathrm{opt}}

\newcommand{\costs}{\mathrm{cost}}

\newcommand{\ifproof}{\iffalse}
\newcommand{\erb}{$\rho$-distance}
\newcommand{\erbb}{$r$-Lipschitz function}

\newcommand{\argmin}{\mathrm{argmin}}

\newcommand{\KM}{\textsc{$k$-Means}}

\newcommand{\KMP}{\textsc{$k$-Means++}}
\newcommand{\KLM}{\textsc{Clustering++}}

\newcommand{\KLMo}{\textsc{$k$-Line-Means++}}
\newcommand{\CPC}{\textsc{$k$ $j$-Subspace-Coreset}}

\providecommand{\mb}[1]{\mathbf#1}
\makeatletter
\newcommand*\bigcdot{\mathpalette\bigcdot@{.5}}
\newcommand*\bigcdot@[2]{\mathbin{\vcenter{\hbox{\scalebox{#2}{$\m@th#1\bullet$}}}}}
\makeatother

\newcommand{\scal}{0.50}

\newcommand{\comment}[1]{}
\newcommand{\REAL}{\ensuremath{\mathbb{R}}}

\newcommand{\abs}[1]        {\left| #1\right|}

\newcommand{\KSC}{\emph{$(\eps, \mbs(j))$-coreset }}

\newcommand{\haf}{\frac{9}{10}}
\newcommand{\iap}{with probability at least }

  \newcounter{theorem}
 \newtheorem{Corollary}[theorem]{Corollary}
  \newtheorem{proof}[theorem]{Proof}
 \newtheorem{Lemma}[theorem]{Lemma}		
 \newtheorem{Theorem}[theorem]{Theorem}		

\newtheorem{Definition}[theorem]{Definition}		

\usepackage{etoolbox}
\usepackage{fullpage}

\begin{document}
\title{Faster Projective Clustering Approximation of Big Data}
\date{\today}

\author{
Adiel Statman\thanks{Computer Science Department, University of Haifa. 
E-mail: \email{statman.adiel@gmail.com}}
\and
Liat Rozenberg\thanks{School of Information and Communication Technology, Griffith University, Australia. 
E-mail: \email{liatle@gmail.com}}
\and
Dan Feldman\thanks{Computer Science Department, University of Haifa. 
E-mail: \email{dannyf.post@gmail.com}}
}

\maketitle
\begin{abstract}
In projective clustering we are given a set of n points in $R^d$ and wish to cluster them to a set $S$ of $k$ linear subspaces in $R^d$ according to some given distance function. An $\eps$-coreset for this problem is a weighted (scaled) subset of the input points such that for every such possible $S$ the sum of these distances is approximated up to a factor of $(1+\eps)$.    
We suggest to reduce the size of existing coresets by suggesting the first $O(\log(m))$ approximation for the case of $m$ lines clustering in $O(ndm)$ time, compared to the existing $\exp(m)$ solution. We then project the points on these lines and prove that for a sufficiently large $m$ we obtain a coreset for projective clustering. Our algorithm also generalize to handle outliers. Experimental results and open code are also provided.
\end{abstract}

\section{Introduction}

\paragraph{\textbf{Clustering and \KM}} For a given similarity measure, clustering is the problem of partitioning a given set of objects into groups, such that objects in the same group are more similar to each other, than to objects in the other groups. There are many different clustering techniques, but probably the most prominent and common technique is Lloyd’s
algorithm or the \KM{ }algorithm~\cite{lloyd1982least}. The input to the classical Euclidean \KM{ }optimization problem is a set $P$ of $n$ points in $\mathbb{R}^d$, and the goal is to group the $n$ points into $k$ clusters, by computing a set of $k$-centers (also points in $\mathbb{R}^d$) that minimizes the sum of squared distances between each input point to its nearest center. The algorithm is initialized with $k$ random points (centroids). At each iteration, each of the input points is classified to its closest centroid. A new set of $k$ centroids is constructed by taking the mean of each of the current $k$ clusters. This method is repeated until convergence or until a certain property holds.
\KMP{ } was formulated and proved in \cite{arthur2007k}. It is an algorithm for a constant bound of optimal $k$-means clustering of a set. Both \KM{ }and \KMP{ } were formulated using the common and relatively simple metric function of sum of squared distances. However, other clustering techniques might require a unique and less intuitive metric function. 
In \cite{statman2020k} we proved bounding for more general metric functions, \erb{ }. One of the many advantages of \erb{ }is that these metrics generalize the triangle inequality~\cite{braverman2016new}. Also note that this set includes the metric function used for $\KM$ and $\KMP$.\\
In this paper we focus on \erbb{ }, which are \erb{ } functions in which $\rho$ is a function of $r$. 

\paragraph{\textbf{SVD}} The Singular Value Decomposition (SVD) was developed by different mathematicians in the 19th
century (see \cite{stewart1993early} for a historical overview). Numerically stable algorithms to compute it
were developed in the 60's \cite{golub1965calculating,golub1971singular}. In the recent years, very fast
computations of the SVD were suggested. The $k$-SVD of an $n\times d$ real matrix $P$ is used to compute its low-rank approximation, which is the projection of the rows of $P$ onto a linear (non-affine) $k$-dimensional subspace that minimizes its sum of squared distances over these rows, i.e.,
\[
\argmin_{X\in\REAL^{d\times k},X^TX=I } \norm{P - PXX^T}^2_F.
\]
Projection of a matrix on a subspace is called a low rank approximation.


\paragraph{\textbf{Coresets}}

For a huge amount of data, Clustering and subspace projection algorithms/solvers are time consuming. Another problem with such algorithms/solvers is that we may not be able to use them for big data on standard machines, since there is not enough memory to provide the relevant computations. 

A modern tool to handle this type of problems, is to compute a data summarization for the input that is sometimes called \emph{coresets}. Coresets also allow us to boost the running time of those algorithms/solvers while using less memory. 

Coresets are especially useful to (a) learn unbounded streaming data that cannot fit into main memory, (b) run in parallel on distributed data among thousands of machines, (c) use low communication between the machines, (d) apply real-time computations on the device, (e) handle privacy and security issues, (f) compute constrained optimization on a coreset that was constructed independently of these constraints and of curse boost there running time.

\paragraph{\textbf{Coresets for SVD}}
In the context of the $k$-SVD problem, given $\eps\in(0,\frac{1}{2})$, an $\eps$-coreset for a matrix $A\in\REAL^{n\times d}$ is a matrix $C\in \REAL^{m\times d}$ where $m \ll n$, which guarantees that the sum of the squared distances from any linear (non-affine) $k$-dimensional subspace to the rows of $C$ will be approximately equal to the  sum of the squared distances from the same $k$-subspace to the rows of $P$, up to a $(1\pm\eps)$ multiplicative factor. I.e., for any matrix $X\in\REAL^{d\times k}$, such that $X^TX=I$ we have,
$$ \abs{\norm{P - PXX^T}^2_F -  \norm{C - CXX^T}^2_F} \leq \eps\norm{P - PXX^T}^2_F.$$
Algorithms that compute $(1+\eps)$-approximation for low-rank approximation and subspace approximation are usually
based on randomization and significantly reduce the running time compared to computing
the accurate SVD \cite{clarkson2009numerical,clarkson2017low,deshpande2006adaptive,deshpande2010efficient, deshpande2011algorithms, feldman2010coresets, nguyen2009fast, sarlos2006improved, shyamalkumar2012efficient}. More
information on the large amount of research on this field can be found in \cite{halko2011finding} and \cite{mahoney2011randomized}.
Indeed the most useful subspace is one resulted from the SVD of the data, which is the subspace which gives the minimal least square error from the data. There are coresets which desiged to approximate data for projecting specifically on this subspace. Such are called "weak" coresets. 
However, in this paper deal with "strong" coresets which approximate the data for projecting on any subspace in the same dimension of the data.
 The first coreset for the $k$-dimensional subspace of size that is independent
of both $n$ and $d$, but are also subsets of the input points, was suggested in \cite{feldman2016dimensionality}.
The coreset size is larger but still polynomial in $O(\frac{k}{\eps})$. 
\paragraph{\textbf{Sparse Coresets for SVD}}
In this paper we consider only coresets that are \emph{subset} of their input points, up to a multiplicative weight (scaling). The advantage of such coresets are: (i)they preserved sparsity of the input, (ii)they enable interpretability, (iii) coreset may be used (heuristically) for other problems, (iv)lead less numerical issues that occur when non-exact linear combination of points are used. Following papers aimed to add
this property, e.g. since it preserves the sparsity of the input, easy to interpret, and more
numerically stable. However, their size is larger relating ones which are not a subset of the data; See an elaborated comparison in \cite{feldman2013turning}.
A coreset of size $O(\frac{k}{\eps^2})$ that is a
bit weaker (preserves the spectral norm instead of the Frobenius norm) but still satisfies our
coreset definition was suggested by Cohen, Nelson, and Woodruff in \cite{cohen2015optimal}. This coreset
is a generalization of the breakthrough result by Batson, Spielman, and Srivastava \cite{batson2012twice}
that suggested such a coreset for $k=d-1$. Their motivation was graph sparsification, where
each point is a binary vector of 2 non-zeroes that represents an edge in the graph. An open
problem is to reduce the running time and understand the intuition behind this result.

Applying Reduction algorithm of \cite{} on our coreset made it appropriate not only for one non-affine subspace, but for projective clustering over any affine $k$-subspaces.

\paragraph{\textbf{NLP Application}}

One  idea behind minimizing the squared Euclidean distance of lexical data such as document-term to the nearest subspace,
 is that the important information of the input
points/vectors lies in their direction rather than their length, i.e., vectors pointing in the same
direction correspond to the same type of information (topics) and low dimensional subspaces
can be viewed as combinations of topics describe by basis vectors of the subspace. For
example, if we want to cluster webpages by their TFIDF (term frequency inverse document
frequency) vectors that contain for each word its frequency inside a given webpage divided
by its frequency over all webpages, then a subspace might be spanned by one basis vector
for each of the words “computer”,“laptop”, “server”, and “notebook”, so that the subspace
spanned by these vectors contains all webpages that discuss different types of computers.


\subsection{Our contribution}
In this chapter we use the problem of $k$-line means, i.e. clustering among $k$ lines which intersect the origin, where will be used formulate a coreset for projective clustering on $k$-$j$ non-affine subspaces. We begin with formulating the distance function that reflect the distance of a point from such line, by comparing it to the measurement of the distance of the projection of this point on a unit sphere to the intersection points of that line with the unit sphere. 
We justify that by bounding this distance function by the distance function of a point to a line. Then we prove that this distance function is indeed a $\rho$-distance and thus the result of Chapter \ref{rodis} can be used in order to bound an optimal clustering among $k$-lines that intersect the origin. Say we sampled $m'$ lines, in that way we get a linear time algorithm which provide a $O(\log(m'))$-approximation for optimal projection on such $m'$-lines. Then we produce a coreset for projective clustering, by sampling such lines with our seeding algorithm, until the sum of distances of the data points from the lines is less than the sum of distances of the data points from the $k$ $j$-dimensional subspaces, and bounds its size depending on an indicator of the data degree of clustering, which is not required to be known a-priory.  
In this paper we: 
\begin{enumerate}
\item Prove a linear time $O(log(k))$-approximation of optimal $k$ non-affine $j$-dimensional subspace of any data in $\REAL^d$.
\item Prove a coreset for any $k$ non-affine $j$-dimensional subspaces received directly by sampling lines that intersect the origin (non-affine).
\item Provide extensive experimental results of our coreset method for a case of one $j$-dimensional subspace, i.e. a coreset for $SVD$ versus and upon the Algorithm of \cite{cohen2015optimal} and provide its pseudo code.
\item Provide full open code. 
\end{enumerate}

\section{$K$-line Clustering}
In this section we define the algorithm $\KLMo$ for approximating a data set $P$ of $n$ points in $\mathbb{R}^d$ by $k$ lines that intersect the origin. The algorithm uses $\KLM$; See Algorithm \ref{1}, with a function $w:P\to[0,\infty)$ and the function $f_{\ell}$ as defined in Definition~\ref{k-lines metric} below. The pseudocode of the algorithm is presented in Algorithm~\ref{2a}. We use this result in order to provide a linear time $O(\log k)$-approximation to clustering over $k$ $j$-subspaces; See Theorem \ref{badconst}.

\begin{Definition}\label{subspa0}
Let $Q\subseteq P$. For every point $p\in P$ let $\pi(p,Q)\in\mathrm{argmin}_ {q\in Q}f(p,q)$. We denote $f(p,Q)=f(p,\pi(p,Q))$. 
\end{Definition}

\begin{Definition}[$\cost$, $\opt$ and partition over a set]\label{def:costopt}
For an integer $m$ let $[m]=\br{1,\cdots,m}$. For a subset $X\subseteq P$ and a point $p\in P$, we denote $f(p,X)=\min_{x\in X}f(p,x)$ if $X\neq\emptyset$, and $f(p,\emptyset)=1$ otherwise.
Given a function $w:P\to(0,\infty)$, for every $G\subseteq P$ we define
\[
\cost(G,w,X):=\disum_{p\in G}w(p)f(p,X).
\]
For an integer $k\in [|G|]$, we define
\[
\opt(G,w,k):=\min_{X^*\subseteq G, |X^*|=k}\cost(G,w,X^*).
\]
Note that we denote $\cost(.,.)=\cost(.,w,.)$ and $\opt(.,.)=\opt(.,w,.)$ if $w$ is clear from the context.\\
A partition $\br{P_1,\cdots, P_k}$ of $P$ over a set $X=\br{x_1..,x_k}\subseteq \REAL^d$ is the partition of $P$ such that for every $i\in[k]$,
$f(p,X)=f(p,x_i)$ for every $p\in P_i$.

A partition $\br{P_1^*,\cdots, P_k^*}$ of $P$ is \emph{optimal} if there exists a subset $X^*\subseteq P$ where $|X^*|=k$ , such that
\[
\disum_{i=1}^k \cost(P_i^*, X^*)=\opt(P, k).
\]
The set $X^*$ is called a \emph{$k$-means} of $P$.
\end{Definition}

\begin{Definition}\label{subspa}
For integers $k>0$ and $j\in[d]$, we denote $\mbs(j)$ to be the union over every possible $j$-dimensional subspace.
\end{Definition}

\begin{Definition}\label{subspa1}
Let $k>0$ and $j\in[d]$ be integers and let $f_0:P^2\to[0,\infty)$ be a function, such that for every $p,q\in P$, 
\[
f_0(p,q)=\norm{p-q}^2.
\]
We denote,
\[
\cost_0(P, Q)=\disum_{p\in P}w(p)\cdot f_0(p,\pi(p,Q)).
\]
We denote, 
\[
\opt_0(P,k,j)=\inf_{\tilde S\subseteq \mbs(j),|\tilde S|=k}\cost_0(P,\tilde S).
\]
\end{Definition}

\begin{Definition}\label{k-lines metric}
For every $p\in P\setminus\br{0}$, let $\hat p=\frac{p}{\norm{p}}$ and let $w:P\to[0,\infty)$ and $f_{\ell}:P^2\to\REAL$ be functions, such that for every $p,q\in P$, 
\[
f_{\ell}(p,q)=\min\{\norm{\hat p-\hat q}^2,\norm{\hat p+\hat q}^2\}.
\]
We denote,
\[
\cost_\ell(P, Q)=\disum_{p\in P}w(p)\cdot\norm{p}^2f_0(p,\pi(p,Q)).
\]
And, 
\[
\opt_\ell(P,k)=\inf_{\tilde S\subseteq \mbs(1),|\tilde S|=k}\cost_\ell(P,\tilde S).
\]
\end{Definition}

\begin{figure}
\begin{subfigure}[h]{0.5\textwidth}
		
		\centering		
		\includegraphics[scale=0.5]{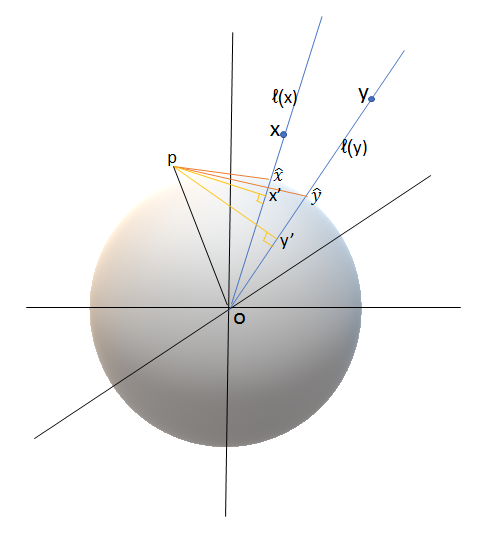}
		\caption{\small \label{klinemeans1}}
\end{subfigure}
\begin{subfigure}[h]{0.5\textwidth}
		\centering		
		\includegraphics[scale=0.5]{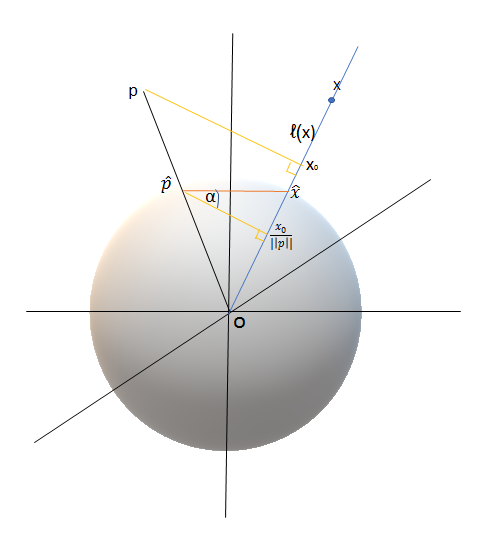}
		\caption{\small \label{klinemeans2}}
\end{subfigure}
\caption{\small \ref{klinemeans1}  An example that shows that if the distance between any point $p\in P$ and a line that intersects the origin $\ell(y)$, is greater than the distance between $p$ and the line that intersects the origin $\ell(x)$, then the distance between $p$ and $\hat y$ is greater than the distance between $p$ and $\hat x$. \ref{klinemeans2} Demonstration of the angle $\alpha$ between the red and yellow lines. As long as $\hat p$ is closer to $\hat x$ than to $-\hat x$, it holds that $0\leq\alpha\leq\frac{\pi}{4}$.}
\end{figure}
\begin{Lemma}\label{2app}
Let $k\geq1$ and $j\in[d]$ be integers and let $\mbs(j)$ be the union over every possible $j$-dimensional subspace in $\REAL^d$. Then,  for every $S\subset\mbs(j)$ such that $|S|=k$ the following hold.
\begin{enumerate}
\renewcommand{\labelenumi}{(\roman{enumi})}
\item For every $p\in\REAL^d$, 
\begin{align}
f_0(p,S)\leq\norm{p}^2f_\ell(p,S)\leq2f_0(p,S).\label{hh1}
\end{align}
\item 
\begin{align}
\cost_0(P,S)\leq\cost_\ell(P,S)\leq2\cost_0(P,S).\label{hh2}
\end{align}
\item 
\begin{align}
\opt_0(P,k,1)\leq\opt_\ell(P,k)\leq2\opt_0(P,k,1).\label{hh3}
\end{align}
\end{enumerate}
\end{Lemma}

\begin{proof}
\begin{enumerate}
\renewcommand{\labelenumi}{(\roman{enumi})}
\item Let $p,x,y\in P$ and let $o$ be the origin. For every $q\in P$ let $\ell(q)$ be the line that intersects the origin and $q$, and see Definition \ref{k-lines metric} for $\hat{q}$.
We will first prove that if $f_0(p,\ell(x))<f_0(p,\ell(y))$, then $f_\ell(p,x)<f_\ell(p,y)$. \\
Without loss of generality we assume that $\norm{\hat{p}-\hat{x}}\leq \norm{\hat{p}+\hat{x}}$, thus
\begin{align}
 f_\ell(p,x)&=\min\br{\norm{\hat{p}-\hat{x}}^2,\norm{\hat{p}+\hat{x}}^2}\nonumber\\
&=\norm{\hat{p}-\hat{x}}^2.\label{loga}
\end{align}
 Let $y'$ be the point in $\ell(y)$ such that $f_0(p,y')= f_0(p,\ell(y))$ (projection) and let $x'$ be the point in $\ell(x)$ such that $f_0(p,x')= f_0(p,\ell(x))$; See Figure \ref{klinemeans1}.
Let $\beta_x$ be the angle between $[o,p]$ and $[o,x]$ and let $\beta_y$ be the angle between $[o,p]$ and $[o,y]$.
Since $\norm{p-x'}^2=f_0(p,x')=f_0(p,\ell(x))<f_0(p,\ell(y))=f_0(p,y')=\norm{p-y'}^2$ we have that also 
$$\norm{p-x'}<\norm{p-y'}.$$ Thus,
\begin{align}
||p||\sin(\beta_x)&=\norm{p-x'}\\
&<\norm{p-y'}\\
&=||p||\sin(\beta_y),
\end{align}
thus $\frac{\pi}{2}\geq\beta_y>\beta_x\geq0$ thus  $\cos(\beta_y)<\cos(\beta_x)$.
We have that
\begin{align}
 f_\ell(p,y)&\geq f_0(\hat{p},\hat{y})\nonumber\\
&=\norm{\hat{p}-\hat{y}}^2\nonumber\\
&=\norm{p}^2+\norm{\hat{y}}^2-2\norm{p}\norm{\hat{y}}\cos(\beta_y)\label{ggt2}\\
&>\norm{p}^2+\norm{\hat{y}}^2-2\norm{p}\norm{\hat{y}}\cos(\beta_x)\nonumber\\
&=\norm{p}^2+\norm{\hat{x}}^2-2\norm{p}\norm{\hat{x}}\cos(\beta_x)\nonumber\\
&=\norm{\hat{p}-\hat{x}}^2\label{ggt3}\\
&= f_0(\hat{p},\hat{x})\nonumber\\
&= f_\ell(p,x),\label{ggt}
\end{align}
where \eqref{ggt2} and \eqref{ggt3} holds bt the Law of cosines and \eqref{ggt} holds by \eqref{loga}.
%
We then get that $f_0(p,\ell(x))<f_0(p,\ell(y))$ yields $f_\ell(p,x)<f_\ell(p,y)$.
For every subspace $S_1\in S$, let $x_0\in S_1$ such that $f_0(p,x_0)=f_0(p,S_1)$, and let $x_\ell\in S_1$ such that $f_\ell(p,x_\ell)=f_\ell(p,S_1)$.
For every $x\in\REAL^d\setminus\br{0}$, let $\ell(x)$ be the line that intersects the origin and $x$; See Figure \ref{klinemeans2}. 

We prove that $x_0\in\ell(x_\ell)$.  
Let us assume by contradiction that $x_0\notin\ell(x_\ell)$. Thus, $f_0(p,\ell(x_0))<f_0(p,\ell(x_\ell))$ and by \eqref{ggt}  $f_\ell(p,x_0)<f_\ell(p,x_\ell)$, 
which contradicts the assumption $f_\ell(p,x_\ell)=f_\ell(p,S)$. Hence we conclude that  $x_0\in\ell(x_\ell)$. 

Therefore,
\begin{align}
\norm{\hat{p}-\hat{x_0}}^2=\min\br{\norm{\hat{p}-\hat{x_\ell}}^2,\norm{\hat{p}+\hat{x_\ell}}^2}. \label{obser}
\end{align}
Without loss of generality we assume that 
\begin{align}
\norm{\hat{p}-\hat{x_\ell}}^2\leq\norm{\hat{p}+\hat{x_\ell}}^2.\label{assu}
\end{align}
Let $\alpha$ be the angle between $\hat{p}-\frac{x_0}{\norm{p}}$ and $\hat{p}-\hat{x_0}$; See Figure \ref{klinemeans2}.Thus $$\cos\alpha=\frac{\norm{\hat{p}-\frac{x_0}{\norm{p}}}}{\norm{\hat{p}-\hat{x_0}}}.$$ 
From \eqref{assu} $0\leq\alpha\leq\frac{\pi}{4}$, so we have that $\frac{1}{2}\leq\cos^2{\alpha}\leq1$. Thus
$$\norm{\hat{p}-\frac{x_0}{\norm{p}}}^2\leq\norm{\hat{p}-\hat{x_0}}^2\leq2\norm{\hat{p}-\frac{x_0}{\norm{p}}}^2.$$Thus we get that
$$\norm{p-x_0}^2\leq\norm{p}^2\norm{\hat{p}-\hat{x_0}}^2\leq2\norm{p-x_0}^2.$$ 
Plugging \eqref{obser} in this yields,
$$\norm{p-x_0}^2\leq\norm{p}^2\min\br{\norm{\hat{p}-\hat{x_\ell}}^2,\norm{\hat{p}+\hat{x_\ell}}^2}\leq\norm{p+x_0}^2.$$ 
We then have that,
$$f_0(p,x_0)\leq \norm{p}^2\cdot f_\ell(p,x_\ell)\leq2f_0(p,x_0),$$
and since $x_\ell\in S_1$ and $x_0\in S_1$ by Definition \ref{subspa} we get that,
$f_0(p,S_1)=f_0(p,x_0)$ and $f_\ell(p,S_1)=f_\ell(p,x_\ell).$
Finally, we obtain
$$f_0(p,S_1)\leq \norm{p}^2\cdot f_\ell(p,S_1)\leq2f_0(p,S_1).$$
Let $S_0\in S$ be a subspace such that $f_0(p,S_0)= f_0(p,S)$, and let $S_\ell\in S$ be a subspace such that $f_\ell(p,S_\ell)= f_\ell(p,S)$. 
We get that,
$$f_0(p,S)=f_0(p,S_0)\leq f_0(p,S_\ell)\leq\norm{p}^2\cdot f_\ell(p,S_\ell)=\norm{p}^2\cdot f_\ell(p,S),$$
and also,
$$\norm{p}^2\cdot f_\ell(p,S)=\norm{p}^2\cdot f_\ell(p,S_\ell)\leq \norm{p}^2\cdot f_\ell(p,S_0)\leq 2f_0(p,S_0)= 2f_0(p,S).$$
Hence,
$$f_0(p,S)\leq \norm{p}^2\cdot f_\ell(p,S)\leq2f_0(p,S).$$
\item Summing \eqref{hh1} over every $p\in P$ and multiplying each side by a weight function $w:P\to[0,\infty)$  we get the result.
\item Let $L_\ell$ be a line such that $\opt_\ell(P,k)=\cost_\ell(P,L_\ell)$ and let $L_0$ be a line such that $\opt_0(P,k,1)=\cost_0(P,L_0)$. From (i) we get that
$$\opt_0(P,k,1)=\cost_0(P,L_0)\leq\cost_0(P,L_\ell)\leq\cost_\ell(P,L_\ell)\leq\opt_\ell(P,k,1).$$ Thus,
\begin{align}
&\opt_0(P,k,1)\leq\opt_\ell(P,k).\label{hl1}
\end{align}
Also we have,
$$\opt_\ell(P,k)=\cost_\ell(P,L_\ell)\leq2\cost_0(P,L_\ell)\leq2\opt_0(P,k,1).$$ Thus,
\begin{align}
&\opt_\ell(P,k)\leq2\opt_0(P,k,1).\label{hl2}
\end{align}
The result follows from \eqref{hl1} and \eqref{hl2}.
\end{enumerate}
\end{proof}

\begin{Definition}[\textbf{\erb{ }function}]\label{rodis}
Let $\rho>0$. A non-decreasing symmetric function $f:P^2\to[0,\infty)$ is an \erb{ } in $P$ if and only if for every $p',q,p\in P$.
\begin{equation}\label{defr}
f(q,p')\leq  \rho \big(f(q,p)+ f(p,p')\big).
\end{equation}
\end{Definition}

\begin{Definition}[$(\rho, \phi,\psi)$ metric\label{def:met}]
Let $(P,f)$ be a $\rho$-metric.
For $\phi,\eps>0$, the pair $(P,f)$ is a $(\rho,\phi,\psi)$-metric if for every $x,y,z\in P$ we have
\begin{equation}\label{fxx}
f(x,z)-f(y,z)\leq \phi f(x,y)+\psi f(x,z).
\end{equation}
\end{Definition}
\newcommand{\pd}{f}
\newcommand{\gr}{P}
\newcommand{\tpd}{g}
\newcommand{\dis}{\mathrm{dist}}		
\begin{Lemma}[Lemma 6 of \cite{statman2020k}\label{leolem}]
Let $\tpd:[0,\infty)\to [0,\infty)$ be a monotonic non-decreasing function that satisfies the following (Log-Log Lipschitz) condition: there is $r>0$ such that for every $x>0$ and $\Delta>1$ we have
\begin{equation}\label{ass3}
\tpd(\Delta x)\leq \Delta^{r}\tpd(x).
\end{equation}
Let $(\gr,\dis)$ be a metric space, and $\pd:\gr^2\to [0,\infty)$ be a mapping from every $p,c\in \gr$ to $\pd(p,c)=\pd(\dis(p,c))$. Then $(P,f)$ is a $(\rho,\phi,\psi)$-metric where
\begin{enumerate}
\renewcommand{\labelenumi}{(\roman{enumi})}
\item  $\rho=\max\br{2^{r-1},1}$,
\item $\phi=\left(\frac{r-1}{\psi}\right)^{r-1}$ and $\psi\in(0,r-1)$, if $r>1$, and
\item $\phi=1$ and $\psi=0$, if $r\leq 1$.
\end{enumerate}
\end{Lemma}

\begin{Lemma}\label{2dist}
The function $f_{\ell}$ is $8-distance$ function in $P$; See Definitions \ref{rodis} and \ref{k-lines metric}.
\end{Lemma}
\begin{proof} 
Let $p,q,y\in P$. Without loss of generality we assume that $\norm{p}\geq\norm{q}$. 
Since for $\tilde{f}_0(x)=x^2$ we have that $\tilde  f_0(\Delta x)\leq\Delta^2\tilde f_0(x)$, from Lemma \ref{leolem} we get that $f_0$ is $2-distance$ function in $P$, see Definition \ref{rodis}.
By Lemma \ref{2app} we have
\begin{align*}
\norm{p}^2f_\ell(p,q)&\leq2f_0(p,q)\\
&\leq4f_0(p,y)+4f_0(y,q)\\
&\leq4f_0(p,y)+4f_0(q,y)\\
&\leq8\norm{p}^2f_\ell(p,y)+8\norm{q}^2f_\ell(q,y)\\
&\leq8\norm{p}^2f_\ell(p,y)+8\norm{p}^2f_\ell(q,y).
\end{align*}
Thus $f_\ell$ is $8-distance$ in $P$.
\end{proof}

\begin{algorithm}
\SetKwInOut{Input}{Input}
\SetKwInOut{Output}{Output}
\caption{\KLM($P, w, X, t, f$); see Theorem~\ref{mainthe}}
\label{1}
\Input{A finite set $P\subseteq\REAL^d$, a function $w:P\to[0,\infty)$, a subset $X\subseteq P$, an integer $t\in[0,|P|-|X|]$ and a function $f:P^2\to[0,\infty)$}.
\Output{$Y\subseteq P$, where $|Y|=|X|+t$.}
$Y:=X$\\
\If{$t\geq1$}{
\For{$i:= 1$ to $t$}
{
For every $p\in P$, $\pr_{i}(p)=\frac{w(p)f(p,Y)}{\disum_{q\in P} w(Y)f(q,Y)}$~\tcp{$f(p,\emptyset):=1.$}
Pick a random point $y_i$ from $P$, where $y_i=p$ with probability $\pr_i(p)$ for every $p\in P$.\label{Lin3}\\
$Y:=X\cup \br{y_1,\cdots,y_i}$
}
}
\Return Y
\end{algorithm}
\begin{Theorem}\label{mainthe}\{Theorem 7 of \cite{statman2020k}\}
Let $P$ be a set of $n$ points in $\REAL^d$ and let $w:P\to[0,\infty)$ be a function. Let $\delta\in(0,1]$ and let $f:P^2\to [0,\infty)$ be a function over $P$.
Let $k\geq2$ be an integer, and $Y$ be the output of a call to $\KLM(P, w, \emptyset, k, f)$; See Algorithm~\ref{1}. Then, with probability at least $1-\delta$,
\[
\cost(P,Y)\leq\frac{8\rho^2}{\delta^2}(1+\ln(k))\opt(P,k).
\]
\end{Theorem}
\begin{algorithm}
\SetKwInOut{Input}{Input}
\SetKwInOut{Output}{Output}
\caption{\textsc{\KLMo}($P, w, X, t$)}
\label{2a}
\Input{A finite set $P$, a function $w:P\to[0,\infty)$, a subset $X\subseteq P$ and an integer $t\in[0,|P|-|X|]$.}
\Output{$Y\subseteq P$, where $|Y|=|X|+t$.}
\Return \KLM($P ,w, X,  t ,  f_{\ell}$) \tcp{See Algorithm~\ref{1} and Definition~\ref{k-lines metric}}
\end{algorithm}
\begin{Theorem}[$k$-line-means' approximation]\label{badconst}
Let $k\geq2$ be an integer and let $[Y,Y']$ be the  output of a call to $\KLMo(P, w,\emptyset, k)$; See Algorithm~\ref{2a}. Let $L$ be a set such that for every $i\in[k]$, the $i$-th element of $L$ is a line that intersect the origin and the $i$-th element of $Y$. 
Then, \iap $1-\delta$,
\[
\cost_0(P,L)\leq \frac{1024}{\delta^2}(1+\ln(k))\opt_0(P,k,1).
\]
Moreover, $L$ can be computed in $O(nkd)$ time.
\end{Theorem}

\begin{proof}
By Lemma~\ref{2dist}, $f_{\ell}$ is $8-distance$ function over $P$. Thus, since in this case $\rho=8$, from Theorem \ref{mainthe} we have that
$\cost_\ell(P,L)\leq \frac{512}{\delta^2}(1+\ln(k))opt_\ell(P,k)$.
Since, by Lemma \ref{2app}, $\cost_0(P,L)\leq\cost_\ell(P,L)$ and $\opt_\ell(P,k)\leq2\opt_0(P,k,1)$, we have that \iap $1-\delta$,
\begin{align*}
\cost_0(P,L)\leq&\cost_\ell(P,L)\\
&\leq\frac{512}{\delta^2}(1+\ln(k))\opt_\ell(P,k)\\
&\leq\frac{1024}{\delta^2}(1+\ln(k))\opt_0(P,k,1).
\end{align*}
\end{proof}

\section{Coresets for projecting on $k$ $j$-subspaces}
In this subsection we use the former results in order to prove an $\eps$-coreset (will be defined below) for projecting on $k$ $j$-subspaces; See Theorem \ref{mainth}.

\begin{Lemma}[\label{col8} Lemma 4 of \cite{statman2020k}]
Let $(P,f)$ be a $(\rho, \phi,\psi)$-metric.
For every set $Z\subseteq P$ we have
\[
|f(x,Z)-f(y,Z)|\leq (\phi+\psi\rho) f(x,y)+\psi\rho \min\br{f(x,Z),f(y,Z)}.
\]
\end{Lemma}

\begin{Corollary}\label{costeq}
Let $C,Q\subseteq P$ and let  $f:P^2\to[0,\infty)$ be a function that holds the conditions of Lemma \ref{leolem}  with every $r>1$. Let $c\in\argmin_{c'\in C}f(p,c')$. Let $\psi\in(0,r-1)$. Then for every $p\in P$ the following holds. 
\begin{align}
&|f(p, Q)-f(c, Q)|\leq((\frac{r-1}{\psi})^{r-1}+\psi\cdot2^{r-1})\cdot f(p, C)+\psi\cdot2^{r-1}\min\{f(p, Q),f(c, Q)\}. \label{loglog2}
\end{align}
\end{Corollary}
\begin{proof}
By plugging values guarnteed in Lemma \ref{costeq} in Lemma \ref{col8} we get the required. 
\end{proof}
\begin{Definition}\label{coreset}
Let $\varepsilon>0$. The set $C\subset\REAL^d$ is called an \KSC for $X$ if for every $S\subseteq\mbs(j)$; See Definition \ref{subspa}, we have
$$|\costs_0(P, S)-\costs_0(C, S)|\leq\eps\cdot\costs_0(P, S).$$
\end{Definition}

\begin{algorithm}
\SetKwInOut{Input}{Input}
\SetKwInOut{Output}{Output}
\caption{\CPC($P, w, a$)}
\label{3}
\Input{A finite set $P$, a function $w:P\to[0,\infty)$ and $a>0$.}
\Output{A tuple of sets $[C,C']$ such that $C\subseteq P$ and $C'_{new}\subset\REAL^d$, such that $|C'_{new}|=|C|$}
$C\gets\KLMo(P, w, \emptyset, 1)$\tcp{See Algorithm \ref{2a}}\label{Line1}
\While{$a\leq\cost_0(P,C,1)$\label{Line2}}
{
$C_{old}\gets C$\\
$C=\KLMo(P,w, C_{old}, 1)$\label{Line3}
}
$C'\gets C$\\
Compute the partition $\br{P_1,\cdots, P_{|Y|}}$ of $P$ over $C$ \tcp{See Definition \ref{def:costopt}}
\For{ every $i\in[|C|]$}
{$u\gets 0$\\
\For{ every $p\in P_i^*$}
{$u\gets u+w(p)$}
$c'_i\gets u\cdot c_i$\tcp{$c_i$ is the $i$-th point of $C$}} 
\Return $[C,C']$
\end{algorithm}

\begin{Theorem}[Coreset for non-affine Projective Clustering]\label{mainth}
Let $P$ be a set of $n$ points in $\REAL^d$ and let $w:P\to[0,\infty)$ be a function. Let $k\geq2$ and $j\in[d]$ be integers, and let $\eps>0$, $\alpha\geq1$ and $\psi\in(0,1)$. Let $\tilde{C}, |\tilde{C}|=k$ be an $\alpha$-approximaion of $\opt(P,k,j)$, i.e $\cost_0(P,\tilde{C},j)\leq\alpha\opt_0(P,k,j)$.
 Let $[C,C']$ be the output of a call to $\CPC(P,w,\eps\cost_0(P,\tilde{C},j))$; See Algorithm~\ref{3} and Definition \ref{subspa1}. 
Then $C'$ is an \emph{$\eps', \mbs(j)$)-coreset} for $k$ clustering by $j$-dimensional subspaces of $P$, where $$\eps'=(\frac{1}{\psi}+2\psi)\eps\alpha+2\psi.$$
Moreover, $C'$ has size $|C'|=O(m^*\log n)$ and can be computed in time $O(ndm^*\log n)$ where
$m^*\in[n]$ is the smallest integer $m$ such that $\opt_0(P,m,1)\leq\eps\alpha\opt_0(P,k,j)$. 
\end{Theorem}

\begin{proof}
For every $p\in P$ let $c_p=\argmin f_0(p,C)$. By Lemma \ref{costeq} we have that for every $S\subseteq\mb{S}(j)$ such that $|S|=k$,
\begin{align}
|\cost_0(P,S)-\cost_0(C',S)|
&=\left|\disum_{p\in P}w(p)\cdot f_0(p, S)-\disum_{p\in P}w(p)\cdot f_0(c_p, S)\right|\nonumber\\
&\leq\disum_{p\in P}w(p)\cdot\left| f_0(p, S)-f_0(c_p, S)\right|\label{dd0}\\
&\leq\disum_{p\in P}w(p)\left((\frac{1}{\psi}+2\psi)f_0(p,C)+2\psi\cdot w(p)f_0(p,S)\right)\label{dd1}\\
&=(\frac{1}{\psi}+2\psi)\disum_{p\in P}w(p)f_0(p,C)+2\psi\disum_{p\in P}w(p)f_0(p,S)\nonumber\\
&=(\frac{1}{\psi}+2\psi)\cost_0(P,C)+2\psi\cost_0(P,S),\label{dd4}
\end{align}
where \eqref{dd0} holds by the triangle inequality, \eqref{dd1} holds by plugging $r=2$ in Corollary \ref{costeq} which holds for $f_0$ since  for $\tilde{f}_0(x)=x^2$ we have that $\tilde  f_0(\Delta x)\leq\Delta^2\tilde f_0(x)$.
Since $C$ is received by $|C'|$ iterations of $\CPC$ algorithm we have that, 
\begin{align}
\cost_0(P, C)&\leq\eps\cost_0(P,\tilde{C},j)\label{fi10}\\
&\leq\eps\alpha\opt_0(P,k,j)\label{fi1i}\\
&\leq\eps\alpha\cost_0(P,S),\label{fii}
\end{align}
where \eqref{fi10} holds by definition of $\tilde{C}$, \eqref{fi1i} holds by the stop condition of the \CPC; See Algorithm \ref{3}. Plugging \eqref{fii} in \eqref{dd4} yields,
\begin{align}
|\costs_0(P, S)-\costs_0(C', S)|&\leq(\frac{1}{\psi}+2\psi)\cdot\eps\alpha\cost_0(P, S)+2\psi\cost_0(P, S).\nonumber\\
&\leq((\frac{1}{\psi}+2\psi)\cdot\eps\alpha+2\psi)\cost_0(P, S).\nonumber
\end{align}

We have that
\begin{align}
\cost(P,C)&\leq\eps\alpha\opt(P,k,j)\label{goodie}\\
&\leq\opt(P,m^*-1,1)\label{dooo},
\end{align}
where \eqref{goodie} holds by Line \ref{Line2} of Algorithm \CPC, and \eqref{dooo} holds by definition of $m^*$.
according to Theorem 11 of \cite{bhattacharya2017k} such inequality holds for $|C|=O(m^*\log n)$ and in time of $O(ndm^*\log n)$.

\end{proof}

\section{Experimental Results}\label{emp}

In this Section we compete the algorithm of \cite{cohen2015optimal} and use it as a coreset for $j$-subspace after our $\KLMo$ pre processing. First let us present the algorithm's lemma and pseudo-code. We call it CNW algorithm.

\begin{Lemma}\label{cnw2}\textsc{(Lemma 11 of \cite{cohen2015dimensionality})}
Let $X$ be a finite set in $\REAL^d$ and let $k\geq1$ be an integer and let $\eps\in(0,\haf]$. Let $C$ be the output of a call to \textsc{CNW}($X, k, \varepsilon$); See Algorithm \ref{4} and \cite{cohen2015optimal}. Then $C$ is an \KSC for $X$ and $|C|=\frac{k}{\eps^2}$. 
\end{Lemma}

\begin{algorithm}
\SetKwInOut{Input}{Input}
\SetKwInOut{Output}{Output}
\caption{\textsc{CNW}($P, k, \varepsilon$)}
\label{4}
\Input{$P\subset\mathbb{R}^{d}$ where $|P|=\ell$, an integer $k\geq1$ and $\varepsilon\in[0,\frac{1}{2}]$}
\Output{$C\subset\mathbb{R}^{d}$ where $|C|=\frac{k}{\eps^2}$}
$\mathbf{UDV^T}\leftarrow$ The $SVD$ of $\mathbf{P}$ \tcp{$\mathbf{P}$ is a $|P|\times d$ matrix in which for every $i\in|P|$, the $i$-th row is the $i$-th point of $P$}
$\mathbf{Q}\leftarrow \mathbf{U_{*,1:k}D_{1:k,1:k}V_{*,1:k}^T}$\\
$\mathbf{Z}\leftarrow{\mathbf{\tilde{V}}_{*,1:2k}}$   (*)\\
$\mathbf{A_2}\leftarrow\frac{\sqrt k}{\norm{\mathbf{P^T-Q}}_F}\cdot(\mathbf{P-ZZ^TP})$\\
$\mathbf{A}\leftarrow\mathbf{A_2|Z}$\\
$\mathbf{X_u}\leftarrow \mathbf{kI}$\\
$\mathbf{X_\ell}\leftarrow \mathbf{-kI}$\\
$\delta_u\leftarrow\varepsilon+2\varepsilon^2$\\
$\delta_\ell \leftarrow\varepsilon-2\varepsilon^2$\\
$\mathbf{r}\leftarrow 0^{\ell\times1}$\\
$\mathbf{Z}\leftarrow 0^{d\times d}$\\
\For{every $i\in[\lceil \frac{k}{\varepsilon^2}]$}
{
$\mathbf{X_u}\leftarrow \mathbf{X_u}+\delta_u\mathbf{A^TA}$\\
$\mathbf{X_\ell}\leftarrow \mathbf{X_\ell}+\delta_u\mathbf{A^TA}$\\
$\mathbf{M_\ell}\leftarrow (\mathbf{Z}-\mathbf{X_\ell})^{-1}$\label{Line11}\\
$\mathbf{M_u}\leftarrow (\mathbf{X_u}-\mathbf{Z})^{-1}$\\
$\mathbf{N_\ell}\leftarrow\mathbf{AM_\ell A^T}$ \label{Line22}\\ 
$\mathbf{N_u}\leftarrow\mathbf{AM_uA^T}$\\
$\mathbf{L}\leftarrow\frac{\mathbf{N_\ell^2}}{\delta_\ell\cdot \mathrm{tr}(\mathbf{N_\ell^2})}-\mathbf{N_\ell}$\label{Line33}\\
$\mathbf{U}\leftarrow\frac{\mathbf{N_u^2}}{\delta_u\cdot \mathrm{tr}(\mathbf{N_u^2})}-\mathbf{N_u}$\\
$j\leftarrow\mathrm{argmax}\bigl(\mathrm{Diag}(\mb{L})-\mathrm{Diag}(\mb{U})\bigr)$\\
$r_j\leftarrow r_j+\frac{1}{\mathbf{U_{jj}}}$\\
$\mathbf{a}\leftarrow \mathbf{A_{j,*}}$\\
$\mathbf{Z} \leftarrow \mathbf{Z}+r_j\mathbf{a^Ta}$
}
$C=\{r_iq_i|r_i\neq0\}_{i=1}^\ell$\\
\Return $C$\\

{ (*)$\mathbf{\tilde{V}_{:,1:2k}}$ is $\mathbf{V_{:,1:2k}}$ or any other $\mathbf{V_{:,1:2k}}$ approximation matrix that holds $\norm{\mathbf{P-ZZ^TP}}^2\leq2\norm{\mathbf{P^T-Q}}^2$ and $\norm{\mathbf{P-ZZ^TP}}_2^2\leq\frac{2}{k}\norm{\mathbf{P^T-Q}}_F^2$}

\end{algorithm}

\begin{algorithm}
\SetKwInOut{Input}{Input}
\SetKwInOut{Output}{Output}
\caption{\textsc{$k$ $j$-Subspace Fixed-size Coreset}($P, w, k, \varepsilon,\opt_0(P,k,j)$)}
\label{5}
\Input{$P\subset\mathbb{R}^{d}$ where $|P|=\ell$, $w:P\to[0,\infty)$, an integer $k\geq1$, $\varepsilon\in(0,1]$ and $\opt_0(P,k,j)$; See Definition \ref{subspa}}
\Output{$C\subset\mathbb{R}^{d}$ where $|C|=\lceil\frac{4k}{\eps^2}\rceil$}
$[Q,\_]=\CPC(P, w, \frac{\varepsilon}{2}\opt_0(P,k,j))$\tcp{See Algorithm \ref{3}}
$C=\textsc{CNW}(Q,k,\frac{\varepsilon}{2})$\tcp{See Algorithm \ref{4}}
\Return $C$\\

\end{algorithm}

We implemented Algorithm \ref{4} and \ref{5} Python 3.6 via the libraries Numpy and Scipy.sparse. We then run experimental results that we summarize in this section.

\subsection{Off-line results}\label{pcaexp}
We use the following three datasets:

(i) Gyroscope data- Have been collected by \cite{anguita2013public} and can be found in \cite{smart}. The experiments have been carried out with a group of 30 volunteers within an age bracket of 19-48 years. Each person performed six activities (WALKING, WALKING UPSTAIRS, WALKING DOWNSTAIRS, SITTING, STANDING, LAYING) wearing a smartphone (Samsung Galaxy S II) on the waist. Using its embedded gyroscope, we captured 3-axial angular velocity at a constant rate of 50Hz. The experiments have been video-recorded to label the data manually.
Data was collected from 7352 measurements, we took the first 1000 points. Each instance consists of measurements from 3 dimensions, $x$, $y$, $z$, each in a dimension of 128.

(ii) Same as (i) but for embedded accelerometer data (3-axial linear accelerations). Again, we took the first 1000 points.

(iii) MNIST test data, first 1000 images. 

\paragraph{Algorithms.}
The algorithms we compared are Uniform sampling, CNW (Algorithm \ref{4}) and Algorithm \ref{5}. In all datasets we ran the experiments with $k=5$ and $k=10$.
\paragraph{Hardware.}
A desktop, with an Intel i7-6850K CPU @ 3.60GHZ
64GB RAM. 
\paragraph{Results. }
We ran those algorithms on those six datasets with different sizes of coresets, between 1000 to 7000, and compared the received error. The error we determined was calculated by the formula $\frac{|\norm{A-AV_AV_A^T}^2-\norm{A-AV_CV_C^T}^2|}{\norm{A-AV_AV_A^T}^2}$, where $A$ is the original data matrix, $V_A$ is the optimal subspace received by SVD on $A$, and $V_C$ is the optimal subspace received by SVD on the received coreset. Results of gyroscope data are presented in Figure \ref{F1} and results of accelerometer data are presented in Figure \ref{F2}.
\paragraph{Discussion. }
 One can notice in Figures \ref{F1} and \ref{F2} the significant differences between the two SVD algorithm than uniform sampling. Relating to times, there is significant difference between our algorithm to CNW. For MNIST; See Figure \ref{F3}, indeed CNW get less error, but also there one should consider to use ours when taking times into account.

 	\begin{figure}[h!]
		\begin{subfigure}[h]{0.33\textwidth}
			\centering
			\includegraphics[width=\textwidth]{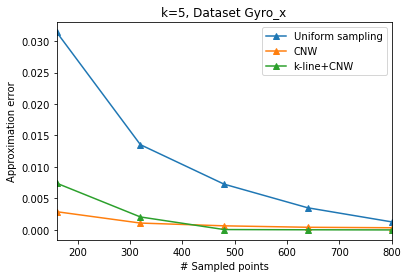}
			\caption{\label{klinemeansa} $x$ error, $k=5$}
		\end{subfigure}
		\begin{subfigure}[h]{0.33\textwidth}
			\centering
			\includegraphics[width=\textwidth]{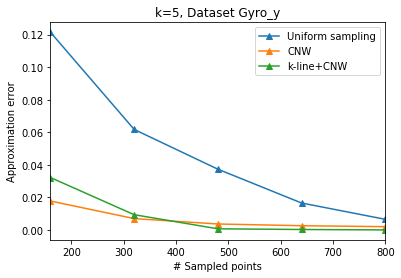}
			\caption{\label{klinemeansa} $y$ error, $k=5$}
		\end{subfigure}
		\begin{subfigure}[h]{0.33\textwidth}
			\centering
			\includegraphics[width=\textwidth]{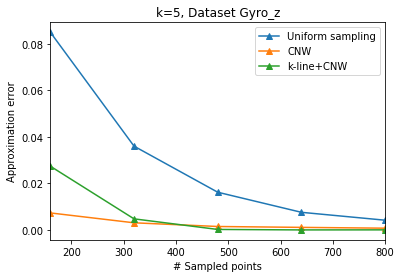}
			\caption{\label{klinemeansa} $z$ error, $k=5$}
		\end{subfigure}
		
		\par\bigskip
		
		\begin{subfigure}[h]{0.33\textwidth}
			\centering
			\includegraphics[width=\textwidth]{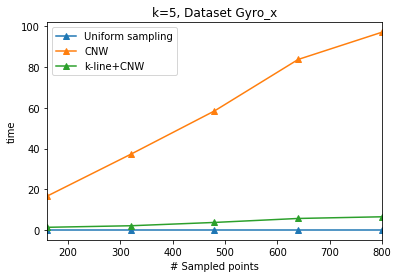}
			\caption{\label{klinemeansa} $x$ times, $k=5$}
		\end{subfigure}
		\begin{subfigure}[h]{0.33\textwidth}
			\centering
			\includegraphics[width=\textwidth]{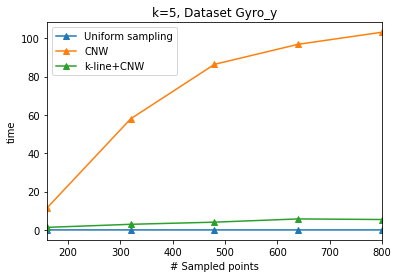}
			\caption{\label{klinemeansa} $y$ times, $k=5$}
		\end{subfigure}
		\begin{subfigure}[h]{0.33\textwidth}
			\centering
			\includegraphics[width=\textwidth]{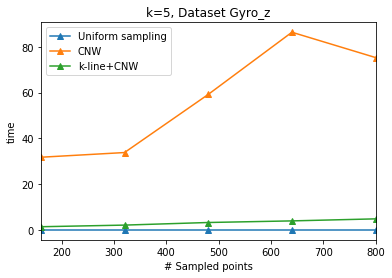}
			\caption{\label{klinemeansa} $z$ times, $k=5$}
		\end{subfigure}
		
		\par\bigskip
		
		\begin{subfigure}[h]{0.33\textwidth}
			\centering
			\includegraphics[width=\textwidth]{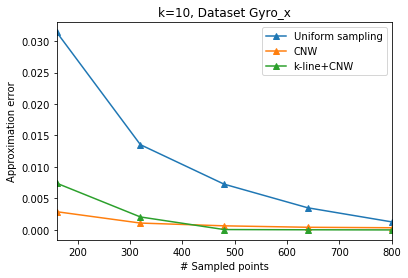}
			\caption{\label{klinemeansa} $x$ error,$k=10$}
		\end{subfigure}
			\begin{subfigure}[h]{0.33\textwidth}
			\centering
			\includegraphics[width=\textwidth]{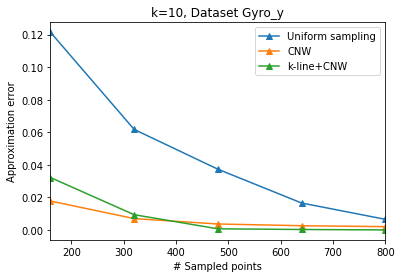}
			\caption{\label{klinemeansa} $y$ error,$k=10$}
		\end{subfigure}
		\begin{subfigure}[h]{0.33\textwidth}
			\centering
			\includegraphics[width=\textwidth]{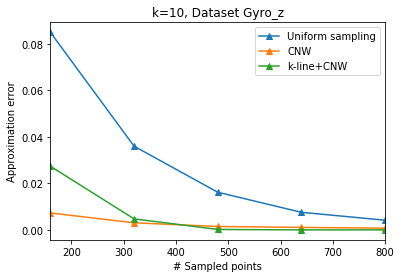}
			\caption{\label{klinemeansa} $z$ error,$k=10$}
		\end{subfigure}

		\par\bigskip

		\begin{subfigure}[h]{0.33\textwidth}
			\centering
			\includegraphics[width=\textwidth]{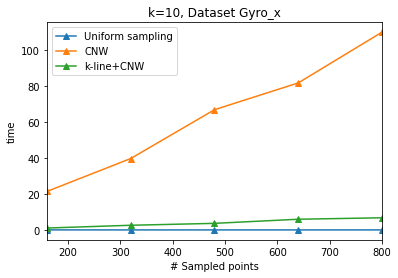}
			\caption{\label{klinemeansa} $x$ times,$k=10$}
		\end{subfigure}
			\begin{subfigure}[h]{0.33\textwidth}
			\centering
			\includegraphics[width=\textwidth]{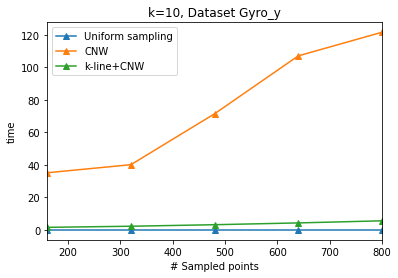}
			\caption{\label{klinemeansa} $y$ times,$k=10$}
		\end{subfigure}
		\begin{subfigure}[h]{0.33\textwidth}
			\centering
			\includegraphics[width=\textwidth]{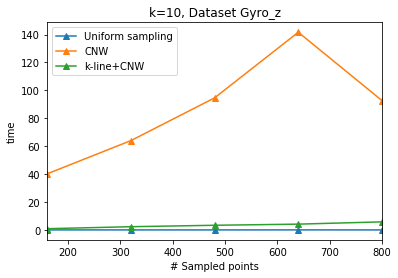}
			\caption{\label{klinemeansa} $z$ times ,$k=10$}
		\end{subfigure}

		\par\bigskip
		
\caption{\small\label{F1} Result of the experiments that are described in Subsection \ref{pcaexp} on gyroscope data for the three sampling algorithms: uniform, CNW(Algorithm \ref{4}) and Algorithm \ref{5}.}
\end{figure}

 	\begin{figure}[h!]
		\begin{subfigure}[h]{0.33\textwidth}
			\centering
			\includegraphics[width=\textwidth]{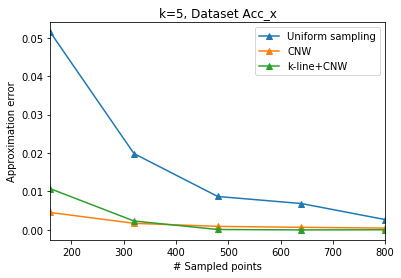}
			\caption{\label{klinemeansa} $x$ error,$k=5$}
		\end{subfigure}
		\begin{subfigure}[h]{0.33\textwidth}
			\centering
			\includegraphics[width=\textwidth]{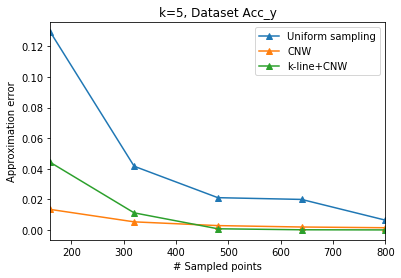}
			\caption{\label{klinemeansa} $y$ error,$k=5$}
		\end{subfigure}
		\begin{subfigure}[h]{0.33\textwidth}
			\centering
			\includegraphics[width=\textwidth]{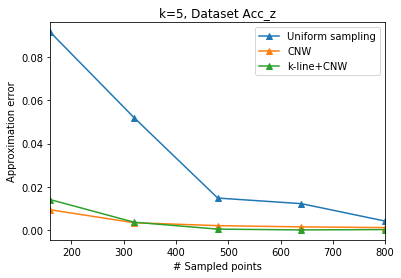}
			\caption{\label{klinemeansa} $z$ error,$k=5$}
		\end{subfigure}
		
		\par\bigskip

		\begin{subfigure}[h]{0.33\textwidth}
			\centering
			\includegraphics[width=\textwidth]{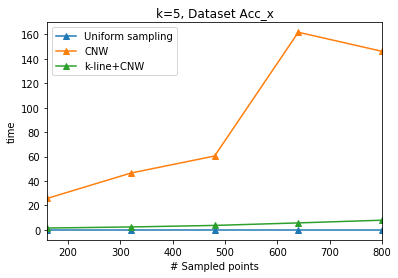}
			\caption{\label{klinemeansa} $x$ times,$k=5$}
		\end{subfigure}
		\begin{subfigure}[h]{0.33\textwidth}
			\centering
			\includegraphics[width=\textwidth]{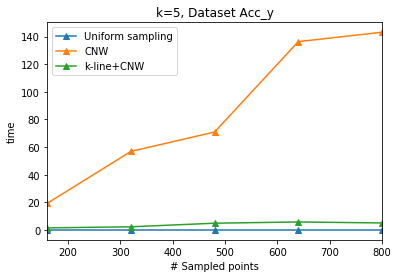}
			\caption{\label{klinemeansa} $y$ times,$k=5$}
		\end{subfigure}
		\begin{subfigure}[h]{0.33\textwidth}
			\centering
			\includegraphics[width=\textwidth]{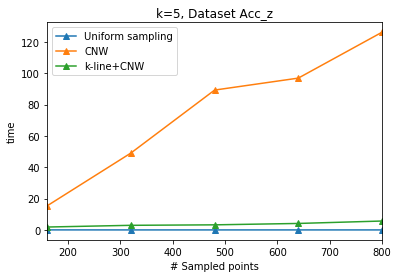}
			\caption{\label{klinemeansa} $z$ times,$k=5$}
		\end{subfigure}

		\par\bigskip

		\begin{subfigure}[h]{0.33\textwidth}
			\centering
			\includegraphics[width=\textwidth]{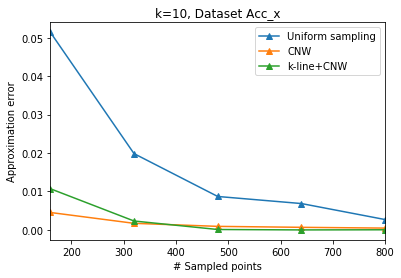}
			\caption{\label{klinemeansa} $x$ error,$k=10$}
		\end{subfigure}
		\begin{subfigure}[h]{0.33\textwidth}
			\centering
			\includegraphics[width=\textwidth]{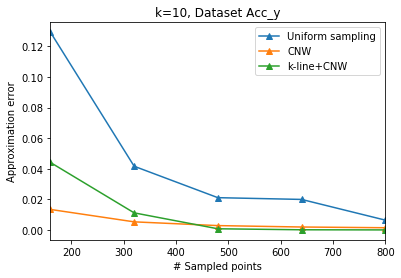}
			\caption{\label{klinemeansa} $y$ error,$k=10$}
		\end{subfigure}
		\begin{subfigure}[h]{0.33\textwidth}
			\centering
			\includegraphics[width=\textwidth]{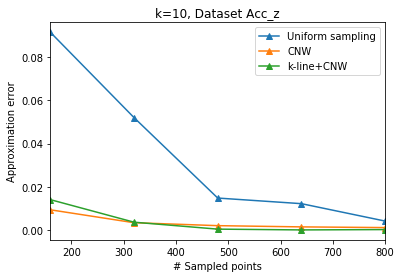}
			\caption{\label{klinemeansa} $z$ error,$k=10$}
		\end{subfigure}

		\par\bigskip

		\begin{subfigure}[h]{0.33\textwidth}
			\centering
			\includegraphics[width=\textwidth]{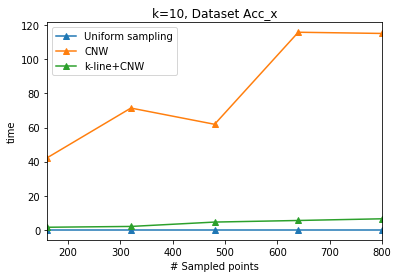}
			\caption{\label{klinemeansa} $x$ times,$k=10$}
		\end{subfigure}
		\begin{subfigure}[h]{0.33\textwidth}
			\centering
			\includegraphics[width=\textwidth]{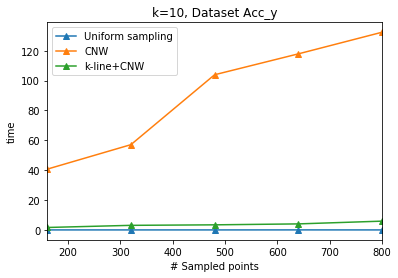}
			\caption{\label{klinemeansa} $y$ times,$k=10$}
		\end{subfigure}
		\begin{subfigure}[h]{0.33\textwidth}
			\centering
			\includegraphics[width=\textwidth]{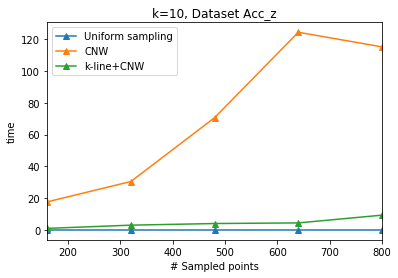}
			\caption{\label{klinemeansa} $z$ times,$k=10$}
		\end{subfigure}
		
		\par\bigskip
		
\caption{\small\label{F2}Result of the experiments that are described in Subsection \ref{pcaexp} on accelerometer data for the three sampling algorithms: uniform, CNW(Algorithm \ref{4}) and Algorithm \ref{5}.}
\end{figure}

 	\begin{figure}[h!]
		\begin{subfigure}[h]{0.5\textwidth}
		\centering
		\includegraphics[width=\textwidth]{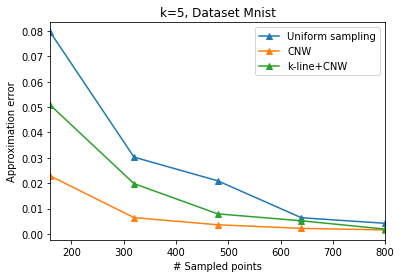}
		\caption{\label{klinemeansa} Error,$k=5$}
	\end{subfigure}
	\begin{subfigure}[h]{0.5\textwidth}
		\centering
		\includegraphics[width=\textwidth]{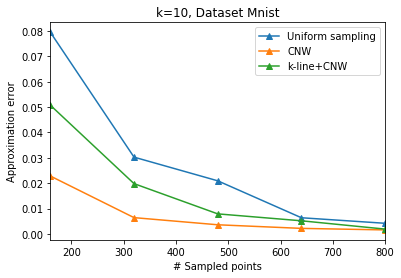}
		\caption{\label{klinemeansa} Error,$k=10$}
	\end{subfigure}
		\begin{subfigure}[h]{0.5\textwidth}
		\centering
		\includegraphics[width=\textwidth]{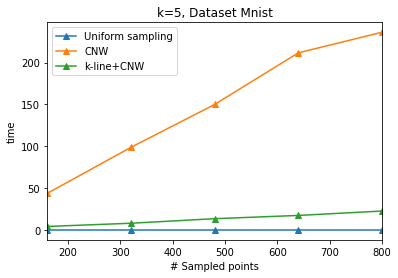}
		\caption{\label{klinemeansa} Time,$k=5$}
	\end{subfigure}
	\begin{subfigure}[h]{0.5\textwidth}
		\centering
		\includegraphics[width=\textwidth]{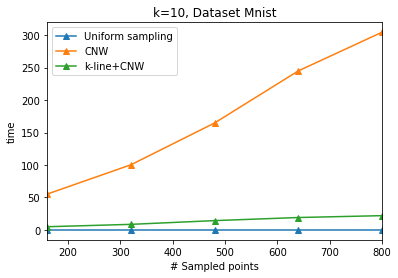}
		\caption{\label{klinemeansa} Time,$k=10$}
	\end{subfigure}

	\par\bigskip

\caption{\small\label{F3}Result of the experiments that are described in Subsection \ref{pcaexp} on MNIST data for the three sampling algorithms: uniform, CNW(Algorithm \ref{4}) and Algorithm \ref{5}.}
\end{figure}

\subsection{Big Data Results}\label{svdexp}

\paragraph{Wikipedia Dataset}
We created a document-term matrix of Wikipedia (parsed enwiki-latest-pages-articles.xml.bz2 from \cite{wic2019}), i.e. sparse matrix with 4624611 rows and 100k columns where each cell $(i,j)$ equals the value of how many appearances the word number $j$ has in article number $i$. We use a standard dictionary of the 100k most common words in Wikipedia found in \cite{dic2012}.

 	\begin{figure}[h!]
		\begin{subfigure}[h]{\scal\textwidth}
		\centering
		\includegraphics[scale=\scal]{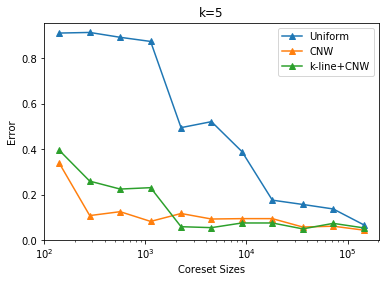}
		\caption{\label{klinemeansa} Error ,$k=5$}
	\end{subfigure}
	\begin{subfigure}[h]{\scal\textwidth}
		\centering
		\includegraphics[scale=\scal]{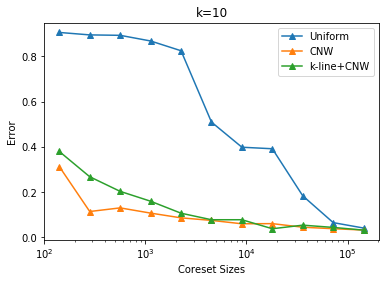}
		\caption{\label{klinemeansa} Error ,$k=10$}
	\end{subfigure}
	\par\bigskip
		\begin{subfigure}[h]{0.5\textwidth}
		\centering
		\includegraphics[scale=0.65]{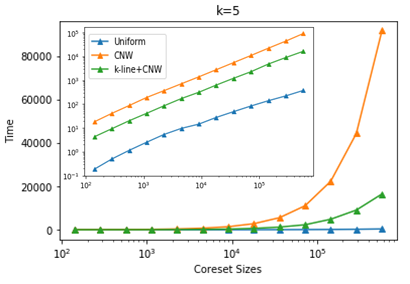}
		\caption{\label{wt3} Times,$k=5$}
	\end{subfigure}
	\begin{subfigure}[h]{0.5\textwidth}
		\centering
		\includegraphics[scale=0.65]{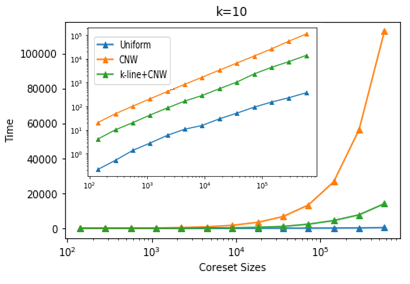}
		\caption{\label{wt4} Times,$k=10$}
	\end{subfigure}
\par\bigskip

\caption{\small\label{F4}Result of the experiments that are described in Subsection \ref{svdexp} on Wikipedia data for the three sampling algorithms: uniform, uniform, CNW(Algorithm \ref{4}) and Algorithm \ref{5}. Time units are [sec]. Small figures within \ref{wt3} and \ref{wt4} are same results as the large and $y$-logged.}
\end{figure}

\paragraph{Our tree system}
Our system separates the $n$ points of the data into chunks of a desired size of coreset, called $m$. It uses consecutive chunks of the data, merge each pair of them, and uses a desired algorithm in order to reduce their dimensionality to a half. The process is described well in \cite{feldman2010coresets}. The result is a top coreset of the whole data.
We built such a system. We used 14 floors for our system, thus divided the n=4624611 points used into 32768 chunks where each chunk, including the top one, is in a size of 141.

\paragraph{Johnson-Lindenstrauss transform}
In order to accelerate the process, one can apply on it a Johnson-Lindenstrauss (JL; see \cite{johnson1984extensions}) transform within the blocks. In our case, we multiplied this each chunk from the BOW matrix by a randomized matrix of $100K$ rows and $d$ columns, and got a dense matrix of $n$ rows as the leaf size and $d$ columns where equals to $k\cdot\log(n)=k\cdot6$,  since analytically proven well-bounded JL transform matrix is of a constant times of $\ln(n)$ (see \cite{johnson1984extensions}) and indeed $\lfloor\ln(141)\rfloor=6$.

\paragraph{Algorithms.}
Same as in Subsection \ref{pcaexp}, the algorithms we compared are Uniform sampling, CNW; See Algorithm \ref{4} and Algorithm \ref{5}.

\paragraph{Hardware.}
Same as in Subsection \ref{pcaexp}
\paragraph{Results. }
We compare the error received for the different algorithms. We show the results in Figure \ref{F4} in $x$-logarithmic scale since the floors' sizes differ multiplicatively.
For every floor, we concatenated the leaves of the floor and measured the error between this subset to the original data. The error we determined was calculated by the formula $$\frac{\norm{A-AV_C^TV_C}^2-\norm{A-AV_A^TV_A}^2|}{\norm{A-AV_A^TV_A}}$$, where $A$ is the original data matrix, $V_A$ received by SVD on A, and $V_C$ received by SVD on the data concatenated in the desired floor. Also here we ran with both $k=5$ and $k=10$.

\paragraph{Discussion. }
 One can notice in Figures \ref{F4} the significant error differences between Uniform sampling, and the coreset techniques. Relating times, one should see from \ref{wt3} and \ref{wt4} that our algorithm is executed in an order og magnitude faster than CNW.
\bibliographystyle{abbrv}
\bibliography{proposal_bib1.bbl}
\end{document}